\documentclass[11pt,a4paper]{article}
\usepackage{geometry}
\usepackage[T1]{fontenc}
\usepackage{lmodern}
\usepackage{microtype,amsmath,amssymb,amsthm,mathtools}
\usepackage{pgfplots}
\pgfplotsset{compat=1.18}
\usepackage[numbers]{natbib}
\usepackage[colorlinks=true,linkcolor=blue!45!black,citecolor=blue!45!black,urlcolor=blue!45!black]{hyperref}
\usepackage{bookmark,needspace}

\hypersetup{pdftitle={Dual lattice attacks for bounded distance decoding, revisited},pdfauthor={Thijs Laarhoven}}
\newcommand{\R}{\mathbb R}

\newcommand{\E}{\mathop{\mathbb E}}
\newcommand{\Prob}{\mathop{\mathbb P}}
\newcommand{\Var}{\operatorname*{Var}}
\newcommand{\Unif}{\operatorname{Unif}}
\newcommand{\vol}{\operatorname{vol}}
\newcommand{\norm}[1]{\lVert#1\rVert}
\newcommand{\Rp}{R_{\lambda_1}}
\newcommand{\cL}{\mathcal{L}}
\newcommand{\dd}{\,\mathrm d}
\newtheorem{theorem}{Theorem}
\newtheorem{lemma}[theorem]{Lemma}
\title{Dual lattice attacks for bounded distance decoding, revisited}
\author{Thijs Laarhoven}
\date{}
\begin{document}
\maketitle
\begin{abstract}
Analyses of dual lattice attacks have often assumed that the individual scores associated with short dual vectors are mutually independent. Laarhoven--Walter used this heuristic to derive explicit trade-offs between the target radius and query time for bounded distance decoding (BDD) with preprocessing. Ducas--Pulles subsequently demonstrated theoretical and experimental failures of this heuristic and proposed an alternative model conditioned on the target norm.

In this note, we prove an explicit asymptotic trade-off for (decision-)BDD with preprocessing in the Haar-random lattice model, without heuristic assumptions. Using moment identities of Siegel and Rogers, we analyze cosine scores over complete dual balls and bound both error probabilities when distinguishing targets planted at a prescribed radius from uniform targets modulo the lattice. Optimizing the dual radius yields a trade-off between target radius and query time that matches the asymptotic prediction from the conditional model of Ducas--Pulles.
\end{abstract}

{\small\noindent\textbf{Keywords:} random lattices; bounded distance decoding; dual lattice attacks; average-case analysis.\par}

%%%%%%%%%%%%%%%%%%%%%%%%%%%%%%%%%%%%%%%%%%%%%%%%%%%%%%%%%%%%%%%%%%%%%%%
%%%%%%%%%%%%%%%%%%%%%%%%%%%%%%%%%%%%%%%%%%%%%%%%%%%%%%%%%%%%%%%%%%%%%%%
%%%%%%%%%%%%%%%%%%%%%%%%%%%%%%%%%%%%%%%%%%%%%%%%%%%%%%%%%%%%%%%%%%%%%%%

\section{Introduction}
\label{sec:intro}

In the context of finding closest lattice points, dual lattice attacks provide a way to recognize targets close to a lattice by testing them against short vectors of its dual. This approach is particularly relevant to bounded distance decoding (BDD), where targets are guaranteed to lie close to the lattice, and to the cryptanalysis of various lattice-based cryptographic schemes based on learning with errors (LWE)~\cite{LW,CMST}. Each dual vector defines a periodic cosine score that equals one when the target is a lattice point, and has expectation zero for uniformly random targets. Short dual vectors can retain a positive average score for targets that lie close to a lattice. Averaging many such scores can therefore reveal a small distinguishing bias. The underlying Fourier viewpoint goes back to Aharonov--Regev~\cite{AR}, who approximate a periodic Gaussian using discrete-Gaussian dual samples.

The setting with preprocessing separates the cost of finding dual vectors from the cost of using them. A list of short dual vectors can be (pre)computed once (for instance, through a lattice sieve) and reused for many target queries. Evaluating the score then takes time essentially linear in the list size, making the required number of dual vectors a natural measure of query complexity. Laarhoven--Walter~\cite{LW} studied this geometric perspective and derived explicit asymptotic trade-offs between the target radius and the query time for BDD and related closest-vector problems. Their analysis suggested that the dual approach is particularly effective when the target is unusually close to the lattice.

The single-target analysis uses an independence heuristic for the individual scores. Under this assumption, concentration estimates predict how many dual vectors are needed to separate a planted target from a uniform one. However, the required independence of the individual scores is not guaranteed. Dependencies between these scores can substantially affect predicted error probabilities. The analysis also uses a Gaussian approximation to the individual inner products. For the exponentially small biases relevant to the attack, this approximation need not preserve the leading exponential term.

Ducas--Pulles~\cite[Sections~4--7]{DP} showed, both theoretically and experimentally, that the independence heuristic can give incorrect score distributions and overly optimistic success probabilities. They proposed a model that fixes the target norm and treats dual vectors as independent samples from a ball. The resulting conditional mean is expressed through a Bessel function, while a central-limit heuristic models fluctuations around it. Their experiments support these predictions, but the analysis still relies on assumptions about conditional score distributions and a ball approximation of the Voronoi cell for uniform targets.

Rigorous analyses are available for other sampling models. Pouly--Shen~\cite{PS} prove guarantees for dual attacks using independent discrete-Gaussian samples, and Qu--Xu~\cite{QX} extend this approach to modulus switching. Bashiri--Wiemers~\cite[Section~4]{BW} likewise study success and sample bounds within a discrete-Gaussian model. For ball-based scores, Carrier et al.~\cite[Section~8]{CDMT} analyze Bessel expressions and false-positive floors, while Debris-Alazard et al.~\cite{DDRT} combine Fourier analysis with lattice averaging to study smoothing.

These developments leave a natural question for BDD with preprocessing: which explicit asymptotic radius/query trade-off follows from the conditional score model, and can that trade-off be established without additional heuristics about the scores? Answering this requires controlling both the mean signal and its fluctuations for planted and uniform targets. 

%%%%%%%%%%%%%%%%%%%%%%%%%%%%%%%%%%%%%%%%%%%%%%%%%%%%%%%%%%%%%%%%%%%%%%%

\subsection{Contributions}

We address this question for Haar-random unimodular lattices, using all nonzero dual vectors in a ball. The planted error has a prescribed norm and is independent of the lattice, and the alternative is a target uniform modulo the lattice. In this setting, classical lattice moment identities provide direct control of the dependent dual-vector sums. With these techniques, we prove an explicit radius/query tradeoff for full short-vector catalogues of Haar-random lattices. The contribution is an achievability guarantee at a prescribed target radius; we do not claim that these complexities are optimal.

\Needspace{23\baselineskip}
\begin{theorem}[Fixed-radius tradeoff]\label{thm:main}
Let $\cL\subset\R^n$ be Haar-random unimodular and let $\Rp$ be the radius of the Euclidean unit-volume ball. For fixed $d\in(0,1)$, define
\begin{equation}\label{eq:curve}
 q(d)=\frac{1-d+\ln((1+d)/2)}{\ln2},\qquad
 \alpha(d)=\mathrm e^d\sqrt{\frac{1-d}{1+d}}.
\end{equation}
Using an initial list $\mathcal S_n$ of the $2^{q(d)n+o(n)}$ shortest nonzero dual vectors $\mathbf{w} \in \mathcal{L}^*$ of norm at most $R=2^{q(d)+o(1)}\Rp$, we can separate a uniform target modulo $\cL$ from a planted target $\mathbf{t}=\mathbf{v}+\mathbf{e}$, where $\mathbf{v}\in\cL$ and $\mathbf{e}$ is uniform on the sphere of radius $\alpha(d)\Rp$. This distinguisher works by computing
\begin{equation}\label{eq:mainscore}
 G(\mathbf{t})=\frac{1}{|\mathcal S_n|}\sum_{\mathbf{w}\in\mathcal S_n}\cos(2\pi\langle\mathbf{w},\mathbf{t}\rangle),\qquad
 \tau=|\mathcal S_n|^{-1/2+o(1)}.
\end{equation}
It outputs ``random'' exactly when $G(\mathbf{t})<\tau$. Both the false-positive and false-negative probabilities are at most $2^{-\Omega(n^{2/3})}$, averaged over the lattice and target.
\end{theorem}

Optimizing the conditional (heuristic) model of Ducas--Pulles~\cite[Corollary~6.4, Heuristic~6.5]{DP} for complete dual balls gives the same exponent. We carry out this asymptotic analysis and prove its achievability; their paper does not state the optimized curve. The resulting curve, illustrating the trade-off between the query time and decoding radius, is shown in Figure~\ref{fig:tradeoff}. For comparison, the heuristic CVP sieve~\cite[Theorem~1]{LaaCVP} has time exponent $\tfrac12\log_2(3/2) \approx 0.2925$. Its horizontal benchmark in Figure~\ref{fig:tradeoff} intersects our curve at $\alpha(d) \approx 0.9423$; below this distance threshold, our distinguishing query exponent is smaller than the time exponent of a direct CVP sieve.

Note that with a known planted radius, the same score can also be used for search-BDD by estimating the error from its gradient and rounding, following the Fourier-decoding approach of~\cite[Sections~3--4]{DRS}. We restrict the formal result here to distinguishing.

\paragraph{The limit of $\alpha \uparrow 1$} As $d\downarrow0$ and $\alpha(d)\uparrow1$, we obtain $q(d)\to\log_2(\mathrm e/2)$. This constant $\mathrm e/2$ already appears in Ducas--Pulles~\cite[Lemma~6.7 and the subsequent discussion]{DP}: their sufficient condition for a positive mean score at target radius $R_\alpha$ is $R R_\alpha<n/(4\pi)\sim(\mathrm e/2)\Rp^2$, obtained from the first Bessel zero. For $R_\alpha\sim\Rp$, the Bessel transition thus occurs at dual radius $R\sim(\mathrm e/2)\Rp$. %More explicitly, by Stirling's formula and $J_x(x)=\Theta(x^{-1/3})$~\cite[Eq.~10.19.8]{DLMF} the normalized signal at the transition satisfies
%\begin{equation}\label{eq:endpointsignal}
% \Gamma(n/2+1)\left(\frac2{n/2}\right)^{n/2} J_{n/2}(n/2)
% =\left(\frac2{\mathrm e}\right)^{n/2+o(n)}.
%\end{equation}
%Its inverse square has the same exponential scale $(\mathrm e/2)^{n+o(n)}$ as the critical dual-ball volume. The endpoint therefore follows from the Bessel signal of prior work; Theorem~\ref{thm:main} proves achievability along the optimized curve for $\alpha<1$.

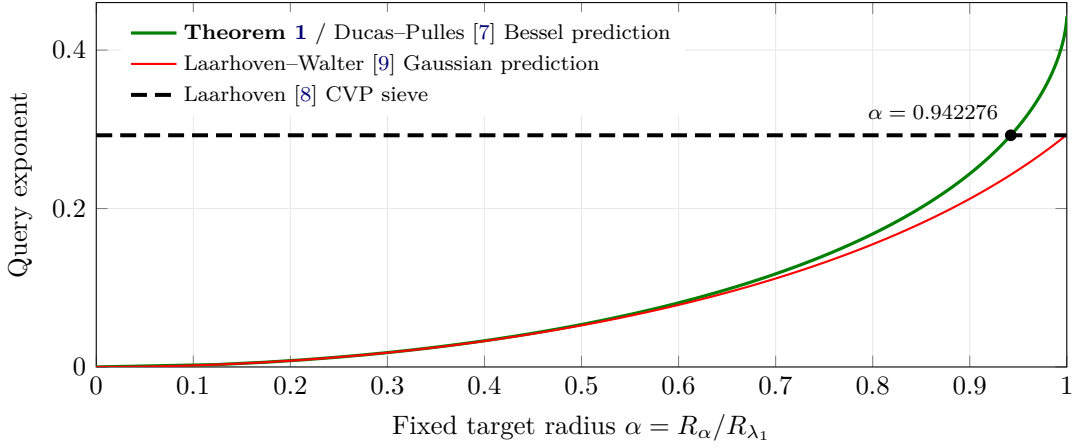
\begin{figure}[!t]
\centering
\begin{tikzpicture}
\begin{axis}[width=0.90\linewidth,height=6.4cm,
 xmin=0,xmax=1,ymin=0,ymax=0.46,
 xlabel={Fixed target radius $\alpha=R_\alpha/\Rp$},ylabel={Query exponent},
 grid=major,grid style={gray!17},
 legend cell align=left,
 legend style={at={(0.025,0.975)},anchor=north west,font=\scriptsize,draw=none,fill=white},
 tick label style={font=\small},label style={font=\small}]
\addplot[green!50!black,very thick] coordinates {
% BEGIN_HAAR_COORDINATES
(0.000000000000,0.000000000000)
(0.123685037437,0.003002479465)
(0.174371864261,0.005998671046)
(0.212895871052,0.008988548408)
(0.245065829370,0.011972085048)
(0.273139731121,0.014949254295)
(0.298279615241,0.017920029312)
(0.321178162523,0.020884383090)
(0.342288251836,0.023842288450)
(0.361925547449,0.026793718038)
(0.380320687562,0.029738644327)
(0.397648400027,0.032677039614)
(0.414044914022,0.035608876019)
(0.429618953948,0.038534125483)
(0.444458991628,0.041452759764)
(0.458638205868,0.044364750442)
(0.472217978257,0.047270068912)
(0.485250421481,0.050168686381)
(0.497780248910,0.053060573874)
(0.509846184088,0.055945702225)
(0.521482041541,0.058824042077)
(0.532717568145,0.061695563883)
(0.543579106951,0.064560237903)
(0.554090127296,0.067418034201)
(0.564271652776,0.070268922645)
(0.574142610193,0.073112872903)
(0.583720116650,0.075949854445)
(0.593019717723,0.078779836538)
(0.602055586537,0.081602788244)
(0.610840691353,0.084418678421)
(0.619386937544,0.087227475720)
(0.627705288623,0.090029148581)
(0.635805869964,0.092823665234)
(0.643698058185,0.095610993697)
(0.651390558528,0.098391101771)
(0.658891472170,0.101163957042)
(0.666208355023,0.103929526875)
(0.673348269303,0.106687778417)
(0.680317828939,0.109438678590)
(0.687123239686,0.112182194093)
(0.693770334704,0.114918291398)
(0.700264606193,0.117646936745)
(0.706611233624,0.120368096148)
(0.712815108995,0.123081735385)
(0.718880859502,0.125787819998)
(0.724812867919,0.128486315294)
(0.730615290996,0.131177186340)
(0.736292076076,0.133860397959)
(0.741846976164,0.136535914733)
(0.747283563603,0.139203700995)
(0.752605242524,0.141863720833)
(0.757815260196,0.144515938080)
(0.762916717397,0.147160316319)
(0.767912577901,0.149796818875)
(0.772805677188,0.152425408817)
(0.777598730432,0.155046048954)
(0.782294339853,0.157658701829)
(0.786895001498,0.160263329723)
(0.791403111490,0.162859894647)
(0.795820971809,0.165448358343)
(0.800150795640,0.168028682280)
(0.804394712328,0.170600827650)
(0.808554771983,0.173164755369)
(0.812632949749,0.175720426069)
(0.816631149784,0.178267800103)
(0.820551208963,0.180806837533)
(0.824394900332,0.183337498135)
(0.828163936335,0.185859741392)
(0.831859971825,0.188373526493)
(0.835484606887,0.190878812329)
(0.839039389478,0.193375557490)
(0.842525817899,0.195863720263)
(0.845945343129,0.198343258630)
(0.849299370994,0.200814130261)
(0.852589264230,0.203276292517)
(0.855816344408,0.205729702441)
(0.858981893754,0.208174316758)
(0.862087156863,0.210610091871)
(0.865133342315,0.213036983859)
(0.868121624204,0.215454948472)
(0.871053143578,0.217863941129)
(0.873929009802,0.220263916914)
(0.876750301852,0.222654830572)
(0.879518069536,0.225036636509)
(0.882233334650,0.227409288782)
(0.884897092079,0.229772741104)
(0.887510310837,0.232126946833)
(0.890073935058,0.234471858972)
(0.892588884936,0.236807430166)
(0.895056057617,0.239133612695)
(0.897476328052,0.241450358474)
(0.899850549806,0.243757619048)
(0.902179555827,0.246055345586)
(0.904464159182,0.248343488880)
(0.906705153757,0.250621999341)
(0.908903314928,0.252890826992)
(0.911059400195,0.255149921468)
(0.913174149794,0.257399232009)
(0.915248287276,0.259638707457)
(0.917282520071,0.261868296252)
(0.919277540010,0.264087946426)
(0.921234023844,0.266297605603)
(0.923152633726,0.268497220990)
(0.925034017680,0.270686739375)
(0.926878810050,0.272866107121)
(0.928687631929,0.275035270164)
(0.930461091567,0.277194174005)
(0.932199784773,0.279342763710)
(0.933904295289,0.281480983900)
(0.935575195156,0.283608778749)
(0.937213045064,0.285726091982)
(0.938818394691,0.287832866863)
(0.940391783021,0.289929046197)
(0.941933738661,0.292014572320)
(0.943444780137,0.294089387098)
(0.944925416184,0.296153431918)
(0.946376146023,0.298206647687)
(0.947797459629,0.300248974821)
(0.949189837989,0.302280353245)
(0.950553753350,0.304300722386)
(0.951889669457,0.306310021166)
(0.953198041791,0.308308187996)
(0.954479317783,0.310295160774)
(0.955733937038,0.312270876876)
(0.956962331536,0.314235273151)
(0.958164925839,0.316188285913)
(0.959342137283,0.318129850941)
(0.960494376164,0.320059903466)
(0.961622045922,0.321978378170)
(0.962725543317,0.323885209174)
(0.963805258599,0.325780330039)
(0.964861575670,0.327663673753)
(0.965894872246,0.329535172729)
(0.966905520012,0.331394758796)
(0.967893884769,0.333242363191)
(0.968860326581,0.335077916557)
(0.969805199912,0.336901348931)
(0.970728853770,0.338712589738)
(0.971631631829,0.340511567789)
(0.972513872565,0.342298211264)
(0.973375909377,0.344072447714)
(0.974218070709,0.345834204050)
(0.975040680166,0.347583406533)
(0.975844056628,0.349319980770)
(0.976628514364,0.351043851704)
(0.977394363133,0.352754943607)
(0.978141908295,0.354453180073)
(0.978871450909,0.356138484007)
(0.979583287835,0.357810777620)
(0.980277711824,0.359469982417)
(0.980955011621,0.361116019193)
(0.981615472046,0.362748808019)
(0.982259374090,0.364368268239)
(0.982886994999,0.365974318456)
(0.983498608358,0.367566876526)
(0.984094484173,0.369145859548)
(0.984674888951,0.370711183854)
(0.985240085778,0.372262765000)
(0.985790334396,0.373800517758)
(0.986325891274,0.375324356103)
(0.986847009687,0.376834193205)
(0.987353939777,0.378329941419)
(0.987846928630,0.379811512276)
(0.988326220340,0.381278816468)
(0.988792056076,0.382731763843)
(0.989244674144,0.384170263390)
(0.989684310051,0.385594223231)
(0.990111196568,0.387003550609)
(0.990525563788,0.388398151876)
(0.990927639185,0.389777932481)
(0.991317647672,0.391142796962)
(0.991695811656,0.392492648929)
(0.992062351094,0.393827391057)
(0.992417483549,0.395146925069)
(0.992761424237,0.396451151729)
(0.993094386085,0.397739970823)
(0.993416579777,0.399013281154)
(0.993728213806,0.400270980520)
(0.994029494522,0.401512965709)
(0.994320626181,0.402739132481)
(0.994601810988,0.403949375554)
(0.994873249148,0.405143588593)
(0.995135138908,0.406321664194)
(0.995387676603,0.407483493870)
(0.995631056698,0.408628968035)
(0.995865471830,0.409757975993)
(0.996091112854,0.410870405917)
(0.996308168880,0.411966144840)
(0.996516827317,0.413045078635)
(0.996717273910,0.414107092000)
(0.996909692783,0.415152068444)
(0.997094266473,0.416179890266)
(0.997271175973,0.417190438545)
(0.997440600767,0.418183593115)
(0.997602718867,0.419159232555)
(0.997757706851,0.420117234166)
(0.997905739897,0.421057473955)
(0.998046991821,0.421979826620)
(0.998181635110,0.422884165524)
(0.998309840959,0.423770362683)
(0.998431779301,0.424638288744)
(0.998547618847,0.425487812965)
(0.998657527113,0.426318803197)
(0.998761670457,0.427131125859)
(0.998860214113,0.427924645925)
(0.998953322219,0.428699226896)
(0.999041157851,0.429454730782)
(0.999123883055,0.430191018078)
(0.999201658880,0.430907947746)
(0.999274645406,0.431605377186)
(0.999343001777,0.432283162220)
(0.999406886232,0.432941157063)
(0.999466456132,0.433579214301)
(0.999521867996,0.434197184868)
(0.999573277525,0.434794918019)
(0.999620839637,0.435372261305)
(0.999664708492,0.435929060550)
(0.999705037525,0.436465159821)
(0.999741979474,0.436980401403)
(0.999775686409,0.437474625773)
(0.999806309762,0.437947671568)
(0.999834000355,0.438399375563)
(0.999858908430,0.438829572637)
(0.999881183675,0.439238095745)
(0.999900975258,0.439624775888)
(0.999918431850,0.439989442084)
(0.999933701658,0.440331921334)
(0.999946932452,0.440652038593)
(0.999958271592,0.440949616736)
(0.999967866061,0.441224476524)
(0.999975862488,0.441476436572)
(0.999982407183,0.441705313315)
(0.999987646160,0.441910920971)
(0.999991725170,0.442093071506)
(0.999994789726,0.442251574597)
(0.999996985138,0.442386237598)
(0.999998456534,0.442496865495)
(0.999999348898,0.442583260875)
(0.999999807091,0.442645223881)
(0.999999975887,0.442682552173)
(1.000000000000,0.442695040889)
% END_HAAR_COORDINATES
};
\addlegendentry{\textbf{Theorem~\ref{thm:main}} / Ducas--Pulles~\cite{DP} Bessel prediction}
\addplot[red,solid,thick] coordinates {
% BEGIN_LW_COORDINATES
(0.000000000000,0.000000000000)
(0.004166666667,0.000003389730)
(0.008333333333,0.000013559112)
(0.012500000000,0.000030508719)
(0.016666666667,0.000054239507)
(0.020833333333,0.000084752814)
(0.025000000000,0.000122050363)
(0.029166666667,0.000166134257)
(0.033333333333,0.000217006987)
(0.037500000000,0.000274671424)
(0.041666666667,0.000339130827)
(0.045833333333,0.000410388839)
(0.050000000000,0.000488449489)
(0.054166666667,0.000573317195)
(0.058333333333,0.000664996763)
(0.062500000000,0.000763493385)
(0.066666666667,0.000868812647)
(0.070833333333,0.000980960526)
(0.075000000000,0.001099943389)
(0.079166666667,0.001225768000)
(0.083333333333,0.001358441518)
(0.087500000000,0.001497971497)
(0.091666666667,0.001644365893)
(0.095833333333,0.001797633061)
(0.100000000000,0.001957781756)
(0.104166666667,0.002124821142)
(0.108333333333,0.002298760784)
(0.112500000000,0.002479610658)
(0.116666666667,0.002667381151)
(0.120833333333,0.002862083060)
(0.125000000000,0.003063727600)
(0.129166666667,0.003272326401)
(0.133333333333,0.003487891514)
(0.137500000000,0.003710435412)
(0.141666666667,0.003939970995)
(0.145833333333,0.004176511590)
(0.150000000000,0.004420070955)
(0.154166666667,0.004670663284)
(0.158333333333,0.004928303205)
(0.162500000000,0.005193005790)
(0.166666666667,0.005464786554)
(0.170833333333,0.005743661459)
(0.175000000000,0.006029646919)
(0.179166666667,0.006322759803)
(0.183333333333,0.006623017438)
(0.187500000000,0.006930437614)
(0.191666666667,0.007245038589)
(0.195833333333,0.007566839089)
(0.200000000000,0.007895858319)
(0.204166666667,0.008232115961)
(0.208333333333,0.008575632183)
(0.212500000000,0.008926427641)
(0.216666666667,0.009284523486)
(0.220833333333,0.009649941366)
(0.225000000000,0.010022703436)
(0.229166666667,0.010402832358)
(0.233333333333,0.010790351309)
(0.237500000000,0.011185283988)
(0.241666666667,0.011587654617)
(0.245833333333,0.011997487953)
(0.250000000000,0.012414809288)
(0.254166666667,0.012839644461)
(0.258333333333,0.013272019860)
(0.262500000000,0.013711962429)
(0.266666666667,0.014159499678)
(0.270833333333,0.014614659685)
(0.275000000000,0.015077471108)
(0.279166666667,0.015547963188)
(0.283333333333,0.016026165759)
(0.287500000000,0.016512109256)
(0.291666666667,0.017005824719)
(0.295833333333,0.017507343806)
(0.300000000000,0.018016698798)
(0.304166666667,0.018533922610)
(0.308333333333,0.019059048795)
(0.312500000000,0.019592111559)
(0.316666666667,0.020133145765)
(0.320833333333,0.020682186945)
(0.325000000000,0.021239271309)
(0.329166666667,0.021804435756)
(0.333333333333,0.022377717881)
(0.337500000000,0.022959155989)
(0.341666666667,0.023548789102)
(0.345833333333,0.024146656972)
(0.350000000000,0.024752800095)
(0.354166666667,0.025367259715)
(0.358333333333,0.025990077843)
(0.362500000000,0.026621297264)
(0.366666666667,0.027260961555)
(0.370833333333,0.027909115090)
(0.375000000000,0.028565803061)
(0.379166666667,0.029231071485)
(0.383333333333,0.029904967221)
(0.387500000000,0.030587537985)
(0.391666666667,0.031278832360)
(0.395833333333,0.031978899816)
(0.400000000000,0.032687790722)
(0.404166666667,0.033405556362)
(0.408333333333,0.034132248953)
(0.412500000000,0.034867921658)
(0.416666666667,0.035612628606)
(0.420833333333,0.036366424908)
(0.425000000000,0.037129366676)
(0.429166666667,0.037901511038)
(0.433333333333,0.038682916161)
(0.437500000000,0.039473641268)
(0.441666666667,0.040273746659)
(0.445833333333,0.041083293728)
(0.450000000000,0.041902344990)
(0.454166666667,0.042730964096)
(0.458333333333,0.043569215860)
(0.462500000000,0.044417166281)
(0.466666666667,0.045274882561)
(0.470833333333,0.046142433138)
(0.475000000000,0.047019887704)
(0.479166666667,0.047907317232)
(0.483333333333,0.048804794004)
(0.487500000000,0.049712391634)
(0.491666666667,0.050630185100)
(0.495833333333,0.051558250771)
(0.500000000000,0.052496666435)
(0.504166666667,0.053445511328)
(0.508333333333,0.054404866171)
(0.512500000000,0.055374813195)
(0.516666666667,0.056355436177)
(0.520833333333,0.057346820476)
(0.525000000000,0.058349053062)
(0.529166666667,0.059362222559)
(0.533333333333,0.060386419277)
(0.537500000000,0.061421735251)
(0.541666666667,0.062468264283)
(0.545833333333,0.063526101980)
(0.550000000000,0.064595345795)
(0.554166666667,0.065676095075)
(0.558333333333,0.066768451099)
(0.562500000000,0.067872517129)
(0.566666666667,0.068988398454)
(0.570833333333,0.070116202442)
(0.575000000000,0.071256038589)
(0.579166666667,0.072408018570)
(0.583333333333,0.073572256293)
(0.587500000000,0.074748867958)
(0.591666666667,0.075937972111)
(0.595833333333,0.077139689703)
(0.600000000000,0.078354144154)
(0.604166666667,0.079581461415)
(0.608333333333,0.080821770034)
(0.612500000000,0.082075201223)
(0.616666666667,0.083341888929)
(0.620833333333,0.084621969908)
(0.625000000000,0.085915583796)
(0.629166666667,0.087222873193)
(0.633333333333,0.088543983738)
(0.637500000000,0.089879064194)
(0.641666666667,0.091228266540)
(0.645833333333,0.092591746050)
(0.650000000000,0.093969661398)
(0.654166666667,0.095362174745)
(0.658333333333,0.096769451844)
(0.662500000000,0.098191662143)
(0.666666666667,0.099628978890)
(0.670833333333,0.101081579252)
(0.675000000000,0.102549644421)
(0.679166666667,0.104033359744)
(0.683333333333,0.105532914843)
(0.687500000000,0.107048503746)
(0.691666666667,0.108580325027)
(0.695833333333,0.110128581940)
(0.700000000000,0.111693482570)
(0.704166666667,0.113275239987)
(0.708333333333,0.114874072401)
(0.712500000000,0.116490203331)
(0.716666666667,0.118123861775)
(0.720833333333,0.119775282392)
(0.725000000000,0.121444705691)
(0.729166666667,0.123132378220)
(0.733333333333,0.124838552778)
(0.737500000000,0.126563488626)
(0.741666666667,0.128307451705)
(0.745833333333,0.130070714875)
(0.750000000000,0.131853558155)
(0.754166666667,0.133656268979)
(0.758333333333,0.135479142458)
(0.762500000000,0.137322481664)
(0.766666666667,0.139186597917)
(0.770833333333,0.141071811094)
(0.775000000000,0.142978449946)
(0.779166666667,0.144906852437)
(0.783333333333,0.146857366092)
(0.787500000000,0.148830348371)
(0.791666666667,0.150826167054)
(0.795833333333,0.152845200650)
(0.800000000000,0.154887838827)
(0.804166666667,0.156954482861)
(0.808333333333,0.159045546114)
(0.812500000000,0.161161454528)
(0.816666666667,0.163302647160)
(0.820833333333,0.165469576730)
(0.825000000000,0.167662710212)
(0.829166666667,0.169882529449)
(0.833333333333,0.172129531811)
(0.837500000000,0.174404230882)
(0.841666666667,0.176707157191)
(0.845833333333,0.179038858988)
(0.850000000000,0.181399903062)
(0.854166666667,0.183790875607)
(0.858333333333,0.186212383148)
(0.862500000000,0.188665053519)
(0.866666666667,0.191149536896)
(0.870833333333,0.193666506910)
(0.875000000000,0.196216661819)
(0.879166666667,0.198800725765)
(0.883333333333,0.201419450105)
(0.887500000000,0.204073614839)
(0.891666666667,0.206764030133)
(0.895833333333,0.209491537949)
(0.900000000000,0.212257013784)
(0.904166666667,0.215061368538)
(0.908333333333,0.217905550516)
(0.912500000000,0.220790547579)
(0.916666666667,0.223717389452)
(0.920833333333,0.226687150209)
(0.925000000000,0.229700950956)
(0.929166666667,0.232759962716)
(0.933333333333,0.235865409552)
(0.937500000000,0.239018571939)
(0.941666666667,0.242220790424)
(0.945833333333,0.245473469586)
(0.950000000000,0.248778082344)
(0.954166666667,0.252136174644)
(0.958333333333,0.255549370556)
(0.962500000000,0.259019377851)
(0.966666666667,0.262547994088)
(0.970833333333,0.266137113281)
(0.975000000000,0.269788733219)
(0.979166666667,0.273504963513)
(0.983333333333,0.277288034451)
(0.987500000000,0.281140306789)
(0.991666666667,0.285064282567)
(0.995833333333,0.289062617114)
(1.000000000000,0.293138132389)
% END_LW_COORDINATES
};
\addlegendentry{Laarhoven--Walter~\cite{LW} Gaussian prediction}
\addplot[black,line width=1.5pt,dash pattern=on 6pt off 3pt,domain=0:1] {0.292481250360578};
\addlegendentry{Laarhoven~\cite{LaaCVP} CVP sieve}
\addplot[only marks,mark=*,mark size=2pt,black,forget plot] coordinates {(0.942275602017,0.292481250361)};
\node[font=\scriptsize,anchor=south east,text=black] at (axis cs:0.94,0.3) {$\alpha=0.942276$};
\end{axis}
\end{tikzpicture}
\caption{Query exponents for a fixed target radius. The solid green curve is our Haar-lattice guarantee and the asymptotic prediction extracted here from the conditional model of Ducas--Pulles~\cite{DP}. The solid red curve is the Laarhoven--Walter Gaussian prediction~\cite[Lemmas~9--10, Eq.~(60)]{LW}. The endpoint at $\alpha=1$ is a limit. The thick dashed black line gives the heuristic time exponent for solving one CVP instance~\cite[Theorem~1]{LaaCVP}.}
\label{fig:tradeoff}
\end{figure}

%%%%%%%%%%%%%%%%%%%%%%%%%%%%%%%%%%%%%%%%%%%%%%%%%%%%%%%%%%%%%%%%%%%%%%%

\subsection*{AI assistance}

ChatGPT was used for mathematical exploration, literature searches, drafting, numerical checks and typesetting. The author takes full responsibility for the final content.

%%%%%%%%%%%%%%%%%%%%%%%%%%%%%%%%%%%%%%%%%%%%%%%%%%%%%%%%%%%%%%%%%%%%%%%
%%%%%%%%%%%%%%%%%%%%%%%%%%%%%%%%%%%%%%%%%%%%%%%%%%%%%%%%%%%%%%%%%%%%%%%
%%%%%%%%%%%%%%%%%%%%%%%%%%%%%%%%%%%%%%%%%%%%%%%%%%%%%%%%%%%%%%%%%%%%%%%

\section{Preliminaries}
\label{sec:prelim}

Write $B_R$ for the Euclidean ball of radius $R$ centered at the origin and $V_R=\vol(B_R)=(R/\Rp)^n$. The dual list is $\mathcal S_n=(\cL^*\cap B_R)\setminus\{\mathbf{0}\}$. We assume $n\ge3$ and denote the Riemann zeta function by $\zeta(s)=\sum_{m=1}^{\infty}m^{-s}$ for $s>1$. Write $\mathrm{Haar}$ for the Haar probability law on unimodular lattices in $\R^n$, and $\Unif(A)$ for the uniform law on $A$. The dual of a Haar-random lattice has the same distribution.

%%%%%%%%%%%%%%%%%%%%%%%%%%%%%%%%%%%%%%%%%%%%%%%%%%%%%%%%%%%%%%%%%%%%%%%
\paragraph{Lattice moment identities}
For a compactly supported bounded measurable function $f:\R^n\to\R$, write $\widehat f(\cL)=\sum_{\mathbf{w}\in \cL^*\setminus\{\mathbf{0}\}}f(\mathbf{w})$. Siegel's identity~\cite{Siegel} gives $\E_{\cL\sim\mathrm{Haar}}\widehat f(\cL)=\int_{\R^n} f(\mathbf{x})\dd\mathbf{x}$. Rogers' second-moment identity~\cite{Rogers}, in the form stated in~\cite[Corollary~1]{PS26}, gives
\begin{equation}\label{eq:rogers}
 \Var_{\cL\sim\mathrm{Haar}}\widehat f(\cL)=\frac1{\zeta(n)}\sum_{p=1}^{\infty}\sum_{q=1}^{\infty}\int_{\R^n}
 \bigl(f(p\mathbf{x})f(q\mathbf{x})+f(p\mathbf{x})f(-q\mathbf{x})\bigr)\dd\mathbf{x}.
\end{equation}
\begin{lemma}[A uniform Rogers bound]\label{lem:rogers}
Let $f:\R^n\to\R$ be measurable, supported on $B_R$ and satisfy $|f(\mathbf{x})|\le A$. Then
\begin{equation}\label{eq:rogersbound}
 \Var_{\cL\sim\mathrm{Haar}}\widehat f(\cL)\le C_nA^2V_R,\qquad
 C_n=4\frac{\zeta(n-1)}{\zeta(n)}-2=2+O(2^{-n}).
\end{equation}
\end{lemma}
\begin{proof}
Each of the two integrals for a pair $(p,q)$ in~\eqref{eq:rogers} is bounded in absolute value by $A^2V_R/\max(p,q)^n$. There are $2m-1$ positive integer pairs with $\max(p,q)=m$, so the sum of absolute values is at most
\begin{equation}\label{eq:rogersseries}
 \frac{2A^2V_R}{\zeta(n)}\sum_{m=1}^{\infty}\frac{2m-1}{m^n}=C_nA^2V_R.
\end{equation}
This proves the variance bound and the absolute convergence of the double series in~\eqref{eq:rogers} for $n\ge3$.
\end{proof}

We will use Lemma~\ref{lem:rogers} with $A=1$ in two ways. For $f=\mathbf1_{B_R}$, the sum $\widehat f(\cL)$ is the list size $|\mathcal S_n|$, with mean $V_R$ and variance at most $C_nV_R$. Chebyshev's inequality therefore gives, for every $\varepsilon>0$,
\begin{equation}\label{eq:count}
 \Prob_{\cL\sim\mathrm{Haar}}\bigl[\bigl||\mathcal S_n|-V_R\bigr|\ge\varepsilon V_R\bigr]\le\frac{C_n}{\varepsilon^2V_R}.
\end{equation}
For $f(\mathbf{w})=\mathbf1_{B_R}(\mathbf{w})\cos(2\pi\langle\mathbf{w},\mathbf{t}\rangle)$, the sum is the numerator of $G(\mathbf{t})$. Its variance is also at most $C_nV_R$ for each fixed $\mathbf{t}$ independent of $\cL$. This second application will bound the probability that a planted target has a score below the threshold.

%%%%%%%%%%%%%%%%%%%%%%%%%%%%%%%%%%%%%%%%%%%%%%%%%%%%%%%%%%%%%%%%%%%%%%%
\paragraph{Ball Fourier transform and Bessel formulas}
Let $J_\nu$ denote the Bessel function of the first kind of order $\nu$. Slicing the ball and using its integral representation gives
\begin{equation}\label{eq:mu}
 \mu(\norm{\mathbf{t}}):=\frac1{V_R}\int_{B_R}e^{2\pi i\langle\mathbf{w},\mathbf{t}\rangle}\dd\mathbf{w}
 =\frac{\Gamma(n/2+1)J_{n/2}(2\pi R\norm{\mathbf{t}})}{(\pi R\norm{\mathbf{t}})^{n/2}},\qquad \mu(0)=1.
\end{equation}
This classical ball transform~\cite[Fact~4.9]{DDRT} is the conditional signal in~\cite[Corollary~6.4]{DP}; it does not average over target radii. If $j_{n/2,j}$ are the positive zeros of $J_{n/2}$, the Bessel product~\cite[Eq.~10.21.15]{DLMF} gives
\begin{equation}\label{eq:besselproduct}
 \mu(\norm{\mathbf{t}})=\prod_{j=1}^{\infty}\left(1-\frac{\norm{\mathbf{t}}^2}{a_j^2}\right),\qquad
 a_j=\frac{j_{n/2,j}}{2\pi R}.
\end{equation}
In particular, $\mu$ is positive and decreasing before its first zero, and $j_{n/2,1}/(n/2)\to1$. For every fixed $z\in(0,1)$, Debye's expansion~\cite[Eq.~10.19.3]{DLMF} gives
\begin{equation}\label{eq:besselasymp}
 \ln\bigl(J_{n/2}(nz/2)\bigr)=\frac n2\left(\sqrt{1-z^2}+\ln\left(\frac{z}{1+\sqrt{1-z^2}}\right)\right)
 +O(\log n).
\end{equation}
The estimate also holds when $z$ varies with $n$ while staying a fixed positive distance from both $0$ and $1$. Stirling's formula gives $\ln\bigl(\Gamma(n/2+1)\bigr)=\frac n2\ln(n/2)-\frac n2+O(\log n)$ and $\Rp^2=\frac{n}{2\pi\mathrm e}(1+O(\log n/n))$. Substituting into~\eqref{eq:mu}, with $z=4\pi R\norm{\mathbf{t}}/n\in(0,1)$ and $\delta=\sqrt{1-z^2}$, yields the signal asymptotic
\begin{equation}\label{eq:muasymp}
 \mu(\norm{\mathbf{t}})=\exp\left(\frac n2\left[\delta-1+\ln\left(\frac2{1+\delta}\right)\right]+O(\log n)\right).
\end{equation}
Thus $\mu$ is exponentially small for each fixed $z\in(0,1)$, with the exponent shown above.

For a uniform target modulo $\cL$ and $\mathbf{w}\in\cL^*$, we will use the elementary identity
\begin{equation}\label{eq:uniformtarget}
 \E_{\mathbf{t}\sim\Unif(\R^n/\cL)}e^{2\pi i\langle\mathbf{w},\mathbf{t}\rangle}=\mathbf1_{\{\mathbf{w}=\mathbf{0}\}}.
\end{equation}
Uniform measure on the Voronoi cell represents this torus measure.

%%%%%%%%%%%%%%%%%%%%%%%%%%%%%%%%%%%%%%%%%%%%%%%%%%%%%%%%%%%%%%%%%%%%%%%
%%%%%%%%%%%%%%%%%%%%%%%%%%%%%%%%%%%%%%%%%%%%%%%%%%%%%%%%%%%%%%%%%%%%%%%
%%%%%%%%%%%%%%%%%%%%%%%%%%%%%%%%%%%%%%%%%%%%%%%%%%%%%%%%%%%%%%%%%%%%%%%

\section{False-positive and false-negative bounds}
\label{sec:finite}

Fix $R>0$ and a target radius $R_\alpha>0$ satisfying $2\pi RR_\alpha<j_{n/2,1}$. Then $\mu$ is positive and decreasing on $[0,R_\alpha]$. We use the list $\mathcal S_n$ and score $G$ from Theorem~\ref{thm:main}. For a nonempty list, symmetry under $\mathbf{w}\mapsto-\mathbf{w}$ gives the equivalent score formula, with the threshold chosen as follows:
\begin{equation}\label{eq:score}
 G(\mathbf{t})=\frac1{|\mathcal S_n|}\sum_{\mathbf{w}\in\mathcal S_n}e^{2\pi i\langle\mathbf{w},\mathbf{t}\rangle},\qquad
 \tau=\frac{V_R\mu(R_\alpha)}{2|\mathcal S_n|}.
\end{equation}
The test outputs ``planted'' if $G(\mathbf{t})\ge\tau$ and ``random'' otherwise. If the list is empty, set $G=0$ and output ``random''. Periodicity gives $G(\mathbf{v}+\mathbf{e})=G(\mathbf{e})$ for every $\mathbf{v}\in\cL$.

\begin{lemma}[False positives for uniform targets]\label{lem:fp}
For a target uniform modulo $\cL$, the probability $P_{\rm FP}$ that the test above outputs ``planted'', averaged over the lattice and target, satisfies
\begin{equation}\label{eq:fp}
 P_{\rm FP}\le\frac4{V_R\mu(R_\alpha)^2}.
\end{equation}
\end{lemma}

\begin{proof}
For every fixed lattice with a nonempty list,~\eqref{eq:uniformtarget} gives
\begin{equation}\label{eq:null}
 \E_{\mathbf{t}\sim\Unif(\R^n/\cL)}G(\mathbf{t})=0,\qquad
 \E_{\mathbf{t}\sim\Unif(\R^n/\cL)}G(\mathbf{t})^2=\frac1{|\mathcal S_n|}.
\end{equation}
In the expansion of the square, only pairs $\mathbf{w},-\mathbf{w}$ have nonzero expectation. Markov's inequality bounds the conditional false-positive probability by $1/(|\mathcal S_n|\tau^2)=4|\mathcal S_n|/(V_R^2\mu(R_\alpha)^2)$. The same bound is zero for an empty list. Averaging over $\cL$ and using $\E|\mathcal S_n|=V_R$ proves~\eqref{eq:fp}.
\end{proof}

\begin{lemma}[False negatives for planted targets]\label{lem:fn}
Let $\mathbf{e}$ be chosen independently of $\cL$ and supported on $B_{R_\alpha}$. For a planted target $\mathbf{v}+\mathbf{e}$ with $\mathbf{v}\in\cL$, the probability $P_{\rm FN}$ that the test above outputs ``random'', averaged over the lattice and error, satisfies
\begin{equation}\label{eq:fn}
 P_{\rm FN}\le\frac{4C_n}{V_R\mu(R_\alpha)^2}.
\end{equation}
This includes errors uniform on the sphere of radius $R_\alpha$.
\end{lemma}

\begin{proof}
For each fixed $\mathbf{e}$, apply Siegel's identity and Lemma~\ref{lem:rogers} to $f(\mathbf{w})=\mathbf1_{B_R}(\mathbf{w})\cos(2\pi\langle\mathbf{w},\mathbf{e}\rangle)$, with $A=1$. Since $\widehat f(\cL)=|\mathcal S_n|G(\mathbf{e})$, we obtain
\begin{equation}\label{eq:scoremom}
 \E_{\cL\sim\mathrm{Haar}}\bigl[|\mathcal S_n|G(\mathbf{e})\bigr]=V_R\mu(\norm{\mathbf{e}}),\qquad
 \Var_{\cL\sim\mathrm{Haar}}\bigl[|\mathcal S_n|G(\mathbf{e})\bigr]\le C_nV_R.
\end{equation}
A false negative means $|\mathcal S_n|G(\mathbf{e})<V_R\mu(R_\alpha)/2$. Because $\mu(\norm{\mathbf{e}})\ge\mu(R_\alpha)$, the sum then differs from its mean by at least $V_R\mu(R_\alpha)/2$. Chebyshev's inequality gives~\eqref{eq:fn} for each fixed $\mathbf{e}$, and averaging over the error law proves the claim.
\end{proof}

For comparison with the conditional score model, the same moment identity gives a more precise variance for the numerator:
\begin{equation}\label{eq:variance}
 \Var_{\cL\sim\mathrm{Haar}}\bigl[|\mathcal S_n|G(\mathbf{e})\bigr]=V_R\bigl(1+\mu(2\norm{\mathbf{e}})\bigr)+\Delta_n,\qquad
 |\Delta_n|\le(C_n-2)V_R.
\end{equation}
Indeed, the diagonal terms $p=q$ in~\eqref{eq:rogers} contribute
\begin{equation}\label{eq:diagonal}
 2\int_{B_R}\cos^2(2\pi\langle\mathbf{w},\mathbf{e}\rangle)\dd\mathbf{w}.
\end{equation}
The off-diagonal terms are bounded as in Lemma~\ref{lem:rogers}.

%%%%%%%%%%%%%%%%%%%%%%%%%%%%%%%%%%%%%%%%%%%%%%%%%%%%%%%%%%%%%%%%%%%%%%%
%%%%%%%%%%%%%%%%%%%%%%%%%%%%%%%%%%%%%%%%%%%%%%%%%%%%%%%%%%%%%%%%%%%%%%%
%%%%%%%%%%%%%%%%%%%%%%%%%%%%%%%%%%%%%%%%%%%%%%%%%%%%%%%%%%%%%%%%%%%%%%%

\section{Proof of Theorem~\ref*{thm:main}}
\label{sec:mainproof}

\begin{proof}
Fix $d\in(0,1)$, write $\alpha=\alpha(d)$, and set $R_\alpha=\alpha\Rp$. Use the dual cutoff $R_\beta=\beta\Rp$, so that $V_{R_\beta}=\beta^n$. For $1<\beta<\mathrm e/(2\alpha)$, put
\begin{equation}\label{eq:phi}
 \delta=\sqrt{1-(2\beta\alpha/\mathrm e)^2},\qquad
 \Phi(\beta,\alpha)=\ln\beta+\delta-1+\ln\left(\frac2{1+\delta}\right).
\end{equation}
The Stirling estimate in Section~\ref{sec:prelim} gives $2\pi R_\beta R_\alpha/(n/2)=2\beta\alpha/\mathrm e+O(\log n/n)$. The signal asymptotic~\eqref{eq:muasymp} therefore gives
\begin{equation}\label{eq:debye}
 V_{R_\beta}\mu(R_\alpha)^2=\exp\bigl(n\Phi(\beta,\alpha)+O(\log n)\bigr).
\end{equation}
Also $2\pi R_\beta R_\alpha<j_{n/2,1}$ for all sufficiently large $n$. Thus any fixed $\beta$ with $\Phi(\beta,\alpha)>0$ gives exponentially small distinguishing error by Lemmas~\ref{lem:fp} and~\ref{lem:fn}.

For fixed $\alpha<1$, $\Phi$ is negative at $\beta=1$, tends to $-\ln\alpha>0$ as $\beta\uparrow\mathrm e/(2\alpha)$, and satisfies
\begin{equation}\label{eq:phiderivative}
 \frac{\partial\Phi}{\partial\ln\beta}=\delta>0.
\end{equation}
Thus $\Phi(\beta,\alpha)$ has a unique zero $\beta_0$ in this interval. Solving $\Phi(\beta_0,\alpha)=0$ gives $\beta_0=(1+\delta)\mathrm e^{1-\delta}/2$, where $\delta$ is evaluated at $\beta_0$. Hence $\alpha=\alpha(\delta)$ and $\log_2\beta_0=q(\delta)$ by~\eqref{eq:curve}. Since $\frac{\dd}{\dd\delta}\ln\alpha(\delta)=-\delta^2/(1-\delta^2)<0$, we have $\delta=d$ and $\beta_0=2^{q(d)}$. The infimum exponent is therefore $q(d)$.

To attain the stated exponent, choose
\begin{equation}\label{eq:mainparameters}
 \beta=\beta_0\mathrm e^{n^{-1/3}}=2^{q(d)}\mathrm e^{n^{-1/3}},\qquad R=\beta\Rp,\qquad V_R=\beta^n.
\end{equation}
Taylor-expand $\Phi(\beta,\alpha)$ in $\ln\beta$ about $\ln\beta_0$. Its value there is zero and its derivative is $d$, so $n\Phi(\beta,\alpha)=d\,n^{2/3}+O(n^{1/3})$. Since $\beta\to\beta_0\in(1,\mathrm e/(2\alpha))$, the Bessel estimate remains applicable. Substitution into~\eqref{eq:debye} gives
\begin{equation}\label{eq:margin}
 V_R\mu(R_\alpha)^2=\exp\bigl(d\,n^{2/3}+O(n^{1/3}+\log n)\bigr).
\end{equation}
The signal $\mu(R_\alpha)$ is positive for all sufficiently large $n$. Lemmas~\ref{lem:fp} and~\ref{lem:fn}, together with~\eqref{eq:margin}, bound both error probabilities for the complete list by $\exp(-d\,n^{2/3}+O(n^{1/3}+\log n))$.

Taking $\varepsilon=1/n$ in~\eqref{eq:count} shows that $\bigl||\mathcal S_n|-V_R\bigr|\le V_R/n$ with probability at least $1-C_n n^2/V_R$. For these lattices, the threshold in~\eqref{eq:score} satisfies
\begin{equation}\label{eq:mainthreshold}
 \tau=\frac{V_R\mu(R_\alpha)}{2|\mathcal S_n|}
       =|\mathcal S_n|^{-1/2+o(1)}=2^{-q(d)n/2+o(n)}.
\end{equation}
The asymptotics follow from~\eqref{eq:margin} and $\ln V_R=q(d)n\ln2+n^{2/3}$. In particular, with probability $1-2^{-\Omega(n)}$,
\begin{equation}\label{eq:listsize}
 |\mathcal S_n|=(1+O(n^{-1}))V_R,\qquad
 \tau^{-2}=V_R^{1+o(1)}=|\mathcal S_n|^{1+o(1)}.
\end{equation}
Thus the list size and volume are asymptotically equal, and the inverse squared cutoff has the same leading exponential order. The subexponential margin in~\eqref{eq:margin} ensures the vanishing error probabilities.

For the query bound, take only the $K_n=\lceil2V_R\rceil$ shortest nonzero dual vectors. This contains the whole ball whenever $\bigl||\mathcal S_n|-V_R\bigr|\le V_R/n$. If the cap is reached within the ball, one may stop with an arbitrary output; this adds at most $C_n n^2/V_R=2^{-\Omega(n)}$ to either error probability. Both error probabilities remain $2^{-\Omega(n^{2/3})}$, with input size $K_n$ and query work $K_n\operatorname{poly}(n)$ both $2^{q(d)n+o(n)}$.
\end{proof}

%%%%%%%%%%%%%%%%%%%%%%%%%%%%%%%%%%%%%%%%%%%%%%%%%%%%%%%%%%%%%%%%%%%%%%%
%%%%%%%%%%%%%%%%%%%%%%%%%%%%%%%%%%%%%%%%%%%%%%%%%%%%%%%%%%%%%%%%%%%%%%%
%%%%%%%%%%%%%%%%%%%%%%%%%%%%%%%%%%%%%%%%%%%%%%%%%%%%%%%%%%%%%%%%%%%%%%%

\end{document}